\documentclass{article}

\usepackage{amsmath,amsthm,amssymb,url}
\usepackage{tikz}
\usetikzlibrary{decorations.pathreplacing}
\usetikzlibrary{arrows}

\usepackage{booktabs,float}
\usepackage{fullpage}

\newtheorem{theorem}{Theorem}[section]
\newtheorem{lemma}[theorem]{Lemma}
\newtheorem{corollary}[theorem]{Corollary}
\newtheorem{proposition}[theorem]{Proposition}
\theoremstyle{definition}

\newtheorem{remark}[theorem]{Remark}

\newcommand{\xa}{\star}

\newcommand{\xd}{\circ}
\newcommand{\cd}{\cdot}

\newcommand{\Sij}{\ensuremath{S_{i \rightarrow j}}}

\begin{document}

\title{Shortened regenerating codes}

\author{Iwan M. Duursma}

\date{April 30, 2015}

\maketitle

\begin{abstract}
For general exact repair regenerating codes, the optimal trade-offs between storage size and repair bandwith remain undetermined. Various outer bounds and partial results have been proposed.
Using a simple chain rule argument we identify nonnegative differences between the functional repair and the exact repair outer bounds. One of the differences is then bounded from below by the repair data of a shortened subcode.
Our main result is a new outer bound for an exact repair regenerating code in terms of its shortened subcodes. In general the new outer bound is implicit and depends on the choice of shortened subcodes. For the linear case we obtain explicit bounds. 
\end{abstract}

\section{Introduction}

Regenerating codes were introduced by Dimakis, Godfrey, Wu, Wainwright and Ramchandran \cite{DGWWR10}. Their main application is in large distributed storage systems where they lead to significant savings by optimizing the trade-off between storage size and repair bandwith. In a distributed storage system (DSS) data is stored at $N$ nodes such that it can be recovered from any combination of $k$ nodes. If a node fails it can be rebuilt by retrieving the information needed for its repair from any combination of $d$ other nodes. An encoding scheme realizing these parameters is called an $(N,k,d)$ regenerating code. 

An $(N,k,d)$ code comes with a secondary set of parameters $(B,\alpha,\beta)$. For data of total size $B$, a part of size at most $\alpha$ is stored at a single node, and bandwith between a node and any of the $d$ nodes helping in its repair is limited to $\beta$. The gains in a DSS are obtained by using a bandwith $d \beta$ for the repair of a single node that is possibly larger than its data size $\alpha$ but much smaller than the total data size $B$. The challenge is, given $(N,k,d)$, to optimize the trade-off between the storage $\alpha$ per node and the repair bandwith $\beta$ between nodes in order to store data of size $B$. Constructive solutions that yield lower bounds for $B$, or inner bounds, can be found in \cite{RSK11}, \cite{SRKR12a}, \cite{SRKR12b}, \cite{GPV13}, \cite{VigneshT13}, \cite{Ernvall13}, \cite{ShumHCXL14}, \cite{TianSAVK15}.

Without the added access and repair constraints, $N$ nodes will be able to store data of size $B=N\alpha$. The requirement that data can be recovered from any $k$ nodes reduces this amount to $B \leq k\alpha$. The requirement that a node can be repaired with help from any $d$ other nodes introduces further overhead and reduces the size of the data. A first upper bound that takes into account both access and repair requirements is 
\begin{equation} \label{eqfrk}
B \leq  (k-\ell)\alpha + \binom{\ell}{2} \beta + \ell (d+1-k) \beta, ~~~~0 \leq \ell \leq k.
\end{equation}
The upper bound holds for functional repair and thus for exact repair regenerating codes. In the exact repair scenario it is required that a damaged node be rebuilt to its original form. Functional repair uses the weaker assumption that a node be rebuilt to a form that preserves the functionality of the DSS.  The upper bound (\ref{eqfrk}) is attained in the functional repair scenario \cite{DGWWR10} (using arguments from network coding) but is not optimal for exact repair regenerating codes. This was first shown by Tian \cite{T14} for codes of type $(n=4,k=3,d=3)$. Further results on outer bounds for exact repair are in \cite{SSK13}, \cite{D14}, \cite{PK15}, \cite{Tian15a}.  

In this paper we present a new improved outer bound for exact repair regenerating codes. First we refine the proof for the outer bound (\ref{eqfrk}) using a simple chain rule argument. This exposes several nonnegative error terms. We then focus on one particular error term and as main result we formulate an improved version of the outer bound where this error term is bounded from below. Theorems 3.2 and 4.2 in \cite{D14} describe two other improvements of the outer bound (\ref{eqfrk}). The arguments that are used in \cite{D14} are different from the ones used in this paper. In Appendix \ref{A433} the different improvements are illustrated by three different proofs for the improved outer bound for the case $(n=4,k=3,d=3)$. 
Before we describe the main results in more detail we introduce the notation.
    
\subsection{Notation} \label{S:1}

We will use the entropy terminology to express the various bounds. When the random variable $X$ corresponds to the drawing of a vector, uniformly at random, from a finite vector space $X$ we have $H(X) = \log |X| = \dim X$, for the appropriate choice of base in the logarithm. For subspaces $X, Y \subset V$, the usual dictionary between entropy and dimension includes the relations 
\[
\begin{array}{lclcl}
\text{(joint entropy)} & &H(X,Y) &= &\dim (X+Y) \\
\text{(conditional entropy)} & &H(X|Y) &= &\dim X / (X \cap Y) = \dim (X+Y) / Y  \\
\text{(mutual entropy)} & &I(X;Y) &= &\dim X \cap Y
\end{array}
\]
To an exact repair regenerating code of type $(N,k,d)$ with secondary parameters $(B,\alpha,\beta)$ correspond random variables $M$, $\{ W_j : 1 \leq j \leq n \}$ and $\{ \Sij : 1 \leq i,j \leq n, i \neq j \}$ that satisfy several  entropy constraints. 

The variable $M$ describes the data to be stored at the $N$ nodes and has entropy $H(M)=B$. The variable $W_j$ is a function of $M$ that describes the data stored at node $j$, and the variable $\Sij$ is a  function of $W_i$ that describes the helper information provided by node $i$ to repair node $j$. The entropy constraints are the following.
\[
\begin{array}{lclclcl}
\text{(Storage)} &~~ &H(W_j) = \alpha, &~~~ &H(W_j|M) = 0, &~~~ &H(M|W_J) = 0~~\text{for $|J|\geq k$.}  \\ 
\text{(Repair)} &~~ &H(\Sij) = \beta, &~~~ &H(\Sij|W_i) =0, &~~~ &H(W_j|S_{I \rightarrow j}) = 0~~\text{for $|I| \geq d, j \not \in I$.} 
\end{array}
\]
Here $W_J$ denotes the joint distribution $W_J = (W_j : j \in J)$ and $S_{I \rightarrow j}$ denotes the joint distribution $S_{I \rightarrow j} = ( \Sij : i \in I),$ for $j \not \in I$. 
Assuming uniform distributions for each of the variables, the conditions $H(M)=B,$~$H(W_j) = \alpha$ and $H(\Sij) = \beta$ describe the size of the underlying space for $M, W_i$ and $\Sij$, respectively. 
The access condition $H(M|W_J) = 0$ for $|J|\geq k$ says that the data can be recovered from information stored on any $k$ nodes, and similarly $H(W_j|S_{I \rightarrow j}) = 0$ for $|I| \geq d, j \not \in I$ says that node $j$ can be rebuilt with helper information received from any $d$ remaining nodes.

For a linear regenerating code the above can be restated in terms of generating and parity-check matrices. The generator matrix is a matrix of size $B \times N\alpha$ with $B$ independent rows and $N$ blocks of columns, with $\alpha$ columns in each block.
The variable $M$ corresponds to the columns space of the matrix, the variable $W_j$ to the column space of the $j$th block of $\alpha$ columns, and the variable $\Sij$ to a subspace of $W_i$. The access conditions say that the full column space $M$ is generated by any $k$ of the $N$ subspaces $W_j$, and that $W_j$ is generated by any $d$ of the subspaces $\Sij$, $i \neq j.$
Details for the parity-check matrix of a linear regenerating code are in \cite[Section 2.1]{D14}.

\begin{figure} \label{F1}
\begin{center}
\begin{tikzpicture}
  [scale=0.6,every node/.style={minimum size=2em}]
\draw
node at (3,10.2) {\large $\beta$} node at (4,10.2) {\large $\beta$} node at (5,10.2) {\large $\beta$} node at (6,10.2) {\large $\beta$} node at (7.2,10.2) {$\cdot$} node at (8.2,10.2) {$\cdot$}
node at (3,9.2) {\large $\beta$} node at (4,9.2) {\large $\beta$} node at (5,9.2) {\large $\beta$} node at (6,9.2) {\large $\beta$} node at (7.2,9.2) {$\cdot$} node at (8.2,9.2) {$\cdot$}
node at (3,8) {$\cdot$} node at (4,8) {\large $\beta$} node at (5,8) {\large $\beta$} node at (6,8) {\large $\beta$} node at (7.2,8) {$\cdot$} node at (8.2,8) {$\cdot$}
node at (3,7) {$\cdot$} node at (4,7) {$\cdot$} node at (5,7) {\large $\beta$} node at (6,7) {\large $\beta$} node at (7.2,7) {$\cdot$} node at (8.2,7) {$\cdot$}
node at (3,6) {$\cdot$} node at (4,6) {$\cdot$} node at (5,6) {$\cdot$} node at (6,6) {\large $\beta$} node at (7.2,6) {$\cdot$} node at (8.2,6) {$\cdot$}
node at (3,5) {$\cdot$} node at (4,5) {$\cdot$} node at (5,5) {$\cdot$} node at (6,5) {$\cdot$} node at (7.2,5) {$\cdot$} node at (8.2,5) {$\cdot$}
node at (3,3.8) {$\cdot$} node at (4,3.8) {$\cdot$} node at (5,3.8) {$\cdot$} node at (6,3.8) {$\cdot$} node at (7.2,3.8) {\large $\alpha$} node at (8.2,3.8) {$\cdot$}
node at (3,2.8) {$\cdot$} node at (4,2.8) {$\cdot$} node at (5,2.8) {$\cdot$} node at (6,2.8) {$\cdot$} node at (7.2,2.8) {$\cdot$} node at (8.2,2.8) {\large $\alpha$};

\draw [<->,left](1.6,8.8) -- node[xshift=-4pt] {\large $d+1-n$} (1.6,10.6);
\draw [<->,left](1.6,4.6) -- node[xshift=-4pt] {\large $\ell$} (1.6,8.4);
\draw [<->,left](1.6,2.4) -- node[xshift=-4pt] {\large $n-\ell$} (1.6,4.2);

\draw [<->,above](2.6,11.6) -- node {\large $\ell$} (6.4,11.6);
\draw [<->,above](6.8,11.6) -- node {\large $n-\ell$} (8.6,11.6);

\draw [decorate,decoration={brace,amplitude=6pt},xshift=0pt,yshift=0pt] (9.6,10.6) -- (9.6,8.8) node [black,midway,xshift=12pt,right] {$\begin{array}{l}\sum_{1 \leq j \leq \ell} H(W_j | S_{[1,n] \backslash j \rightarrow j}) \\[2ex] ~~~\leq~ \ell (d+1-n) \beta \end{array}$};
\draw [decorate,decoration={brace,amplitude=6pt},xshift=0pt,yshift=0pt] (9.6,8.4) -- (9.6,4.6) node [black,midway,xshift=12pt,right] {$\sum_{1 \leq i < j \leq \ell} H(S_{i \rightarrow j}) ~\leq~ \binom{\ell}{2} \beta $};
\draw [decorate,decoration={brace,amplitude=6pt},xshift=0pt,yshift=0pt] (9.6,4.2) -- (9.6,2.4) node [black,midway,xshift=12pt,right] {$H(W_{[\ell+1,n]}) ~\leq~ (n-\ell) \alpha $};

\end{tikzpicture}
\end{center}
\caption{Data collection from $n$ nodes.}
\end{figure}

\subsection{Outline and Results}

To describe the main result, consider the data collection scenario in Figure~1. Data is collected from a subset of $n$ nodes that are numbered $1$ to $n$ (out of a total of $N$ nodes). For a given $\ell$ with $0 \leq \ell \leq n$, the contents of nodes $\ell+1$ to $n$ is read from the nodes, an amount of size $(n-\ell)\alpha.$ The contents of nodes $\ell,\ell-1,\ldots,1$ is recovered in that order using repair information. When it is time to collect repair information for node $j$, for $1 \leq j \leq \ell$, repair information for that node is already available from nodes $j+1,\ldots,n$. The missing repair information can be collected from nodes $1,\ldots,j-1$ and form any $d+1-n$ nodes that are not among nodes $1$ to $n.$ Thus, with $B(n)=H(W_{[1,n]})$ the information content of $n$ nodes, 
\begin{equation} \label{eqfr2}
B(n) \leq B_\ell(n) := (n-\ell)\alpha + \binom{\ell}{2} \beta + \ell (d+1-n) \beta, ~~~~0 \leq \ell \leq n.
\end{equation}

In Section \ref{Sref}, we prove a version of this bound that includes an extra error term.

\medskip

(Theorem \ref{T1}) For $0 \leq \ell \leq n$,
\begin{align*}
B(n) + \Delta  &~\leq~ H(W_{[\ell+1,n]}) + \sum_{1 \leq i < j \leq \ell} H(S_{i \rightarrow j}) + \sum_{j \leq \ell} H(W_j | S_{[1,n] \backslash j \rightarrow j}) \label{eq1}  
\end{align*}
where
$\Delta = \sum_{1 \leq i < j \leq \ell} H(S_{i \rightarrow j} | W_{[i+1,n]}) \geq 0$.
\bigskip

In Section \ref{Simp}, we exploit the error term to improve the upper bound.

\medskip

(Theorem \ref{T2}) For $u = 1,2, \ldots,v$, let $W_{n+u}$ be such that
\[ 
H( \Sij | W_{n+u} ) \leq H ( \Sij | W_{[i+1,n+u-1]} ),~~~~~~\text{for $1 \leq i < j \leq \ell$.}
\]
Then
\begin{align*}
B(n) + v B(n) ~\leq~ &H(W_{[\ell+1,n]}) + 
\sum_{1 \leq i < j \leq \ell} H(\Sij) + \sum_{j \leq \ell} H(W_j | S_{[1,n] \backslash j}) + \\
&~~+  \sum_{u=1}^{v} \left( H(W_{[\ell+1,n]}) + H(W_{n+u} | W_{[\ell+1,n]}) 
+ \sum_{j \leq \ell} H(W_j | S_{[1,n] \backslash j} W_{n+u}) \right) 
\end{align*}
\medskip

In Section \ref{Slin},  we give a choice $W_{n+u}$ for linear regenerating codes 
such that $H(W_{n+u}) = u ( n \alpha - B(n) ).$ With this choice the upper bound becomes

\medskip

(Theorem \ref{T3}) For a linear regenerating code, and for $v \geq 0$,
\[
\binom{v+2}{2} B \leq (v+1) B_\ell(n) + \binom{v+1}{2} n \alpha - v \binom{l}{2} \beta. 
\]
For $n=k+1=d+1$, and for $\ell=n$, 
\[
\binom{v+2}{2} B \leq \binom{v+1}{2} n \alpha + \binom{n}{2} \beta.
\]
This bound is Theorem 1.1 in \cite{PK15}. It is attained by layered codes (defined in \cite{TianSAVK15}).

\section{Refinement of the exact repair outer bound} \label{Sref}

For a regenerating code of length $N$ we fix an arbitrary ordering of the $N$ nodes and denote by $B(n)$ the amount of data on the last $n$ nodes. For a given $n$, we number the last $n$ nodes from $1$ to $n$. The remaining $N-n$ nodes are numbered from $0$ downwards.
The outer bound (\ref{eqfr2}) 
is piece-wise linear of the form $B(n) \leq \min \{ B_\ell(n) : 0 \leq \ell \leq n \}$, with each of the $B_\ell(n)$ a linear combination of the storage per node $\alpha$ and the helper bandwith between nodes $\beta$. 
In this section we derive a version $B(n) + \Delta \leq \min \{ B_\ell(n) : 0 \leq \ell \leq n \}$ with an explicit error term~$\Delta$. For $n=k$, the error term $\Delta$ gives a lower bound for the gap between the functional repair and the exact repair outer bounds. 

In deriving the outer bound we will only refer to helper information $\Sij$ for $i < j$. 
For a given $j$, $1 \leq j \leq n$, we consider the sequence $X_1, X_2, \ldots, X_n$ of $n$ variables 
\begin{equation} \label{eqj1}
X_i = \begin{cases} &S_{i \rightarrow j}, ~~\text{for $i < j$.} \\  &W_i, ~~\text{for $i \geq j$.}\end{cases}
\end{equation}
In either of the two cases $X_i$ is a function of the information $W_i$ at node $i$. The following lemma is a straightforward application of the chain rule and holds for an arbitrary sequence of $n$ random variables.
Nonetheless it is at the basis of everything that follows.

\begin{lemma} \label{L1} For a sequence $X_1, \ldots, X_n$ of $n$ random variables, and for $1 \leq j \leq n$, 
\[
H(X_j | X_{[j+1,n]}) + \sum_{i < j} H(X_i | X_{[i+1,n]}) = \sum_{i < j} H(X_i | X_{[i+1,n] \backslash j}) + H(X_j | X_{[1,n] \backslash j})
\]
\end{lemma}
\begin{proof} The claim says that the joint entropy is the same for the sequence $0,X_1, \ldots, X_n$ and the permuted sequence with $0$ and $X_j$ exchanged. For a formal proof, apply the chain rule $j$ times to $H(X_{[1,n]})$ and $j-1$ times to $H(X_{[1,n] \backslash j})$,
\begin{align*}
H(X_{[1,n]}) &= H(X_{[j+1,n]}) + \sum_{i \leq j} H(X_i | X_{[i+1,n]}), \\
H(X_{[1,n] \backslash j}) &= H(X_{[j+1,n]}) + \sum_{i < j} H(X_i | X_{[i+1,n] \backslash j}).
\end{align*}
Now use $H(X_{[1,n]}) = H(X_{[1,n] \backslash j}) + H(X_j | X_{[1,n] \backslash j})$.
\end{proof}

We apply the lemma to the sequence (\ref{eqj1}). 

\begin{proposition} \label{P1}
\[
H(W_j | W_{[j+1,n]}) + \sum_{i < j} H(S_{i \rightarrow j} | W_{[i+1,n]}) \leq \sum_{i < j} H(S_{i \rightarrow j}) + H(W_j | S_{[1,n] \backslash j \rightarrow j})
\]
\end{proposition}
\begin{proof}
In Lemma \ref{L1} we replace the second term on the left with a smaller term and the two terms on the right with larger terms.
\end{proof}

Using the proposition $\ell$ times we obtain a refinement of the outer bound (\ref{eqfr2}). 

\begin{theorem} \label{T1}
For $0 \leq \ell \leq n \leq d+1$,
\begin{align}
B(n) + \Delta  &~\leq~ H(W_{[\ell+1,n]}) + \sum_{1 \leq i < j \leq \ell} H(S_{i \rightarrow j}) + \sum_{j \leq \ell} H(W_j | S_{[1,n] \backslash j \rightarrow j}) \label{eq1} \\ 
&~\leq~ (n-\ell)\alpha + \binom{\ell}{2} \beta + \ell (d+1-n) \beta. \nonumber
\end{align}
where
$\Delta = \sum_{1 \leq i < j \leq \ell} H(S_{i \rightarrow j} | W_{[i+1,n]}) \geq 0$. 
\end{theorem}
\begin{proof}
With $B(n)=H(W_{[1,n]})$, 
\begin{align*}
&B(n) + \sum_{j \leq \ell} \sum_{i < j} H(S_{i \rightarrow j} | W_{[i+1,n]}) \\
=~ &H(W_{[\ell+1,n]}) + \sum_{j \leq \ell} H(W_j | W_{[j+1,n]}) + \sum_{j \leq \ell} \sum_{i < j} H(S_{i \rightarrow j} | W_{[i+1,n]}) \\
\leq~ &H(W_{[\ell+1,n]}) + \sum_{j \leq \ell} \sum_{i < j} H(S_{i \rightarrow j}) 
+  \sum_{j \leq \ell} H(W_j | S_{[1,n] \backslash j \rightarrow j}).
\end{align*}
For the inequality use the proposition.
\end{proof}

The theorem shows that for a given $0 \leq \ell \leq n$, there is a gap in the upper bound (\ref{eqfr2}) of size at least
\begin{equation} \label{eqg1}
\Delta = \sum_{1 \leq i < j \leq \ell} H(S_{i \rightarrow j} | W_{[i+1,n]})
\end{equation}
The terms in the sum capture that part of the helper information $\Sij$ may be redundant. There are two important cases with $\Delta = 0$, Minimum Storage Regenerating codes (MSR codes) and Minimum Bandwith Regenerating codes (MBR codes).  MSR codes have $I(W_j ; W_{J \backslash j}) = 0$ for $|J|=k$ and $H(M)=k\alpha.$
For MSR codes, the bound (\ref{eqfr2}) is achieved for $\ell=0$ and the summation for $\Delta$ is empty. Note that the summation for $\Delta$ is empty also when $\ell = 1$. MBR codes have $I(\Sij ; S_{I \backslash i \rightarrow j}) = 0$ for $|I|=d$ and $H(W_j) = d\beta.$ For MBR codes, the bound (\ref{eqfr2}) is achieved for $\ell=k$ and the terms in the summation for $\Delta$ are all zero.  For values of $\ell \not \in \{0,1,k\}$ we obtain improvements of (\ref{eqfr2}) from lower bounds for the gap (\ref{eqg1}). Our approach is to collect the helper information at a separate node $W$ such that $H(S_{i \rightarrow j} | W) \leq H(S_{i \rightarrow j} | W_{[i+1,n]})$. The same chain rule argument of Lemma \ref{L1} goes through if we add $W$ as node $W_{n+1}$ to the nodes $W_1, W_2, \ldots, W_n.$ This is worked out in the next section.

While Theorem \ref{T1} focuses on $\Delta$ as the main gap in the upper bound (\ref{eqfr2}), three other gaps can be pointed out. They are due to the transition from an equality in Lemma \ref{L1} to an inequality in Proposition \ref{P1} by replacing three of the terms. We quantify these gaps but will not consider them further in this paper.
\begin{align*}
H(X_i | X_{[i+1,n]}) 
&= H(S_{i \rightarrow j} | S_{[i+1,j-1] \rightarrow j} W_{[j,n]} ) \\
&= H(S_{i \rightarrow j} | W_{[i+1,n]}) + C^{(1)}_{i,j} \\[1ex]
C^{(1)}_{i,j} &= I( S_{i \rightarrow j} ; W_{[i+1,j-1]} | S_{[i+1,j-1] \rightarrow j} W_{[j,n]} ) \geq 0 \\[2ex]
H(X_i | X_{[i+1,n] \backslash j}) 
&= H(S_{i \rightarrow j} | S_{[i+1,j-1] \rightarrow j} W_{[j+1,n]} ) \\
&= H(S_{i \rightarrow j}) - C^{(2)}_{i,j} \\[1ex]
C^{(2)}_{i,j} &= I(S_{i \rightarrow j} ; S_{[i+1,j-1] \rightarrow j} W_{[j+1,n]} ) \geq 0 \\[2ex]
H(X_j | X_{[1,n] \backslash j}) 
&= H(W_j |  S_{[1,j-1] \rightarrow j} W_{[j+1,n]} ) \\
&= H(W_j |  S_{[1,n] \backslash j \rightarrow j} W_{[j+1,n]} ) \\
&= H(W_j |  S_{[1,n] \backslash j \rightarrow j}) -  C^{(3)}_j \\[1ex]
C^{(3)}_j &= I(W_j; W_{[j+1,n]} | S_{[1,n] \backslash j}) \geq 0
\end{align*}


All three gaps vanish for the important class of layered codes (defined in \cite{TianSAVK15}).
 

\section{Improvement of the exact repair outer bound} \label{Simp}

Starting point for the outer bound (\ref{eq1}) in Theorem \ref{T1} is Lemma \ref{L1}. The identity
\[
H(X_j | X_{[j+1,n]}) + \sum_{i < j} H(X_i | X_{[i+1,n]}) = \sum_{i < j} H(X_i | X_{[i+1,n] \backslash j}) + H(X_j | X_{[1,n] \backslash j})
\]
holds for any $n$ random variables with a common joint distribution. We applied it $\ell$ times, for $j \leq \ell$. For each $j \leq \ell$
it was used with the choice of variables
\begin{equation} \label{eqj2}
X_i = \begin{cases} &S_{i \rightarrow j}, ~~\text{for $i < j$.} \\  &W_i, ~~\text{for $i \geq j$.}\end{cases}
\end{equation}
To estimate the term $\Delta$ given by (\ref{eqg1}) we apply the same bound $v$ more times. Each time, before the bound is applied we add a carefully chosen term $X_{n+u} = W_{n+u}$ to the sequence $X_1, X_2, \ldots, X_{n+u-1}$, for $1 \leq u \leq v.$ Here $W_{n+u}$ is any function of $M$ such that
\begin{equation} \label{eqW}
H( \Sij | W_{n+u} ) \leq H ( \Sij | W_{[i+1,n+u-1]} ),~~~~~~\text{for $1 \leq i < j \leq \ell$.}
\end{equation}
The variables $W_{n+u}$ may be identified with added virtual nodes. 
For given $j \leq \ell$, the application of Lemma \ref{L1} to the extended sequence
$X_1, X_2, \ldots, X_{n+u}$ yields, for $0 \leq u \leq v$,
\begin{equation} \label{eqv}
H(X_j | X_{[j+1,n+u]}) + \sum_{i < j} H(X_i | X_{[i+1,n+u]}) ~=~ \sum_{i < j} H(X_i |  X_{[i+1,n+u] \backslash j}) + H(X_j | X_{[1,n+u] \backslash j})
\end{equation}
In the following lemma we take the sum of these $v+1$ equations.

\begin{lemma}  \label{L2}
Let $X_1, \ldots, X_n, X_{n+1}, \ldots, X_{n+v}$ be random variables such that 
\[
H(X_i | X_{n+u}) \leq H(X_i | X_{[i+1,n+u-1]}) ~~~~~~\text{for $1 \leq i < \ell$, $1 \leq u \leq v$}.
\]
Then, for $1 \leq j \leq \ell$, 
\[
\sum_{u=0}^v H(X_j | X_{[j+1,n+u]}) + \sum_{i < j} H(X_i | X_{[i+1,n+v]})
~\leq~ \sum_{i < j} H(X_i |  X_{[i+1,n] \backslash j}) + \sum_{u=0}^v H(X_j | X_{[1,n+u] \backslash j})
\]
\end{lemma}
\begin{proof} Take the summation of (\ref{eqv}) over $0 \leq u \leq v$. For $1 \leq u \leq v$, the inequality  
\[
H(X_i |  X_{[i+1,n+u] \backslash j}) \leq H(X_i | X_{n+u}) \leq H(X_i | X_{[i+1,n+u-1]})
\]
gives a cancellation of terms in the summation.
\end{proof}


We apply the lemma with the sequence (\ref{eqj2}).

\begin{proposition} \label{P2}
Let $W_{n+u}$, $1 \leq u \leq v$, be such that
\[ 
H( \Sij | W_{n+u} ) \leq H ( \Sij | W_{[i+1,n+u-1]} ),~~~~~~\text{for $1 \leq i < j \leq \ell$.}
\]
Then, for $1 \leq j \leq \ell$,
\[
\sum_{u=0}^v H(W_j | W_{[j+1,n+u]}) + \sum_{i < j} H(\Sij | W_{[i+1,n+v]})
~\leq~ \sum_{i < j} H(\Sij) + \sum_{u=0}^v H(W_j | S_{[1,n] \backslash j} W_{n+u}).
\]
\end{proposition}
\begin{proof}
In the result of Lemma \ref{L2} we replace the terms on the left with smaller terms and the terms on the right with larger terms. Since $X_{n+u} = W_{n+u}$ and $H( X_i | W_i ) = 0$ for all $i \leq n+u$, the condition in the proposition guarantees the condition that is needed for the lemma. For $1 \leq i \leq \ell$,
\[
H(X_i | X_{n+u}) = H(X_i | W_{n+u} ) \leq H(X_i | W_{[i+1,n+u-1]}) \leq H(X_i | X_{[i+1,n+u-1]}).
\]
\end{proof}

\begin{theorem} \label{T2} For given $0 \leq \ell \leq n$ and $v \geq 0$, let $W_{n+u}$, $1 \leq u \leq v$, be such that
\[ 
H( \Sij | W_{n+u} ) \leq H ( \Sij | W_{[i+1,n+u-1]} ),~~~~~~\text{for $1 \leq i < j \leq \ell.$}
\]
Then
\begin{align*}
B(n) + v B(n) ~\leq~ &H(W_{[\ell+1,n]}) + 
\sum_{1 \leq i < j \leq \ell} H(\Sij) + \sum_{j \leq \ell} H(W_j | S_{[1,n] \backslash j}) + \\
&~~+  \sum_{u=1}^{v} \left( H(W_{[\ell+1,n]}) + H(W_{n+u} | W_{[\ell+1,n]}) 
+ \sum_{j \leq \ell} H(W_j | S_{[1,n] \backslash j} W_{n+u}) \right). 
\end{align*}
\end{theorem}
\begin{proof}
 The proposition yields, after summation over $1 \leq j \leq \ell$,
\begin{align*}
&H(W_{[1,\ell]} | W_{[\ell+1,n]}) + \sum_{u=1}^{v} H(W_{[1,\ell]} | W_{[\ell+1,n]} W_{n+u})  \\
~\leq~
&\sum_{j \leq \ell} \sum_{i < j} H(\Sij) + \sum_{j \leq \ell} H(W_j | S_{[1,n] \backslash j})
+ \sum_{u=1}^{v} \sum_{j \leq \ell} H(W_j | S_{[1,n] \backslash j} W_{n+u}).
\end{align*}
So that
\begin{align*}
B(n) + v B(n) ~\leq~ &
\sum_{j \leq \ell} \sum_{i < j} H(\Sij) + \sum_{j \leq \ell} H(W_j | S_{[1,n] \backslash j})
+ \sum_{u=1}^{v} \sum_{j \leq \ell} H(W_j | S_{[1,n] \backslash j} W_{n+u})\;+\\
&~~+H(W_{[\ell+1,n]}) +  \sum_{u=1}^{v} H(W_{[\ell+1,n]} W_{n+u}). 
\end{align*}
After reordering the terms, the claim follows.
\end{proof}


\begin{corollary} \label{C2} With notation and conditions as in the theorem,
\begin{align*}
(v+1) B(n) ~\leq~ &(v+1) B_\ell(n) +  \sum_{u=1}^{v} \left( H(W_{n+u}) - \binom{\ell}{2} \beta \right). 
\end{align*}
\end{corollary}

\begin{remark} \label{R2}
If we use Theorem \ref{T2} to obtain an upper bound for 
the  conditional entropy $H(W_{[1,\ell]} | W_{[\ell+1,n]})$ we find
\begin{align*}
(v+1) (B(n)-B(n-\ell)) ~\leq~ &(v+1) (B_\ell(n)-(n-\ell)\alpha) +  \sum_{u=1}^{v} \left( H(W_{n+u}) - \binom{\ell}{2} \beta \right). 
\end{align*}
Together with $B(n-\ell) \leq (n-\ell)\alpha$ this gives Corollary \ref{C2}.
\end{remark}

\section{Linear regenerating codes} \label{Slin}

For linear regenerating codes, Condition (\ref{eqW})
\[ 
H( \Sij | W_{n+u} ) \leq H ( \Sij | W_{[i+1,n+u-1]} ),~~~~~~\text{for $1 \leq i < j \leq \ell$,}
\]
holds for
\[
W_{n+u} = \langle \Sij \cap W_{[i+1,n+u-1]} : 1 \leq i < j \leq \ell \rangle
\]
and thus for
\begin{equation} \label{eqW2}
W_{n+u} = \langle W_i \cap W_{[i+1,n+u-1]} : i \leq \ell \rangle.
\end{equation}

\begin{lemma} \label{L3} For $W_{n+u}$ as in (\ref{eqW2}),
\begin{align*}
H(W_{n+u}) &~\leq~  u\; ( \sum_{i=1}^\ell  H(W_i) + H(W_{[\ell+1,n]}) - H(W_{[1,n]}) ) \\
&~\leq~ u\; \min ( \ell \alpha, n \alpha - H(W_{[1,n]}) )
\end{align*}
\end{lemma}
\begin{proof}
\begin{align*}
H(W_{n+u}) &\leq~ \sum_{i \leq \ell} ( H(W_i) - H(W_i | W_{[i+1,n+u-1]}) ) \\
&=~ \sum_{i \leq \ell} H(W_i) + H( W_{[\ell+1,n+u-1]}) - H( W_{[1,n+u-1]}) \\
&\leq~ \sum_{i \leq \ell} H(W_i) + H( W_{[\ell+1,n]}) +H(W_{[n+1,n+u-1}]) - H( W_{[1,n]} ) 
\end{align*}
Now apply induction to complete the proof.
\end{proof}


We apply Corollary \ref{C2}.


\begin{theorem} \label{T3}
For a linear regenerating code, and for $v \geq 0$,
\[
\binom{v+2}{2} B(n) \leq (v+1) B_\ell(n) + \binom{v+1}{2} n \alpha - v \binom{l}{2} \beta.
\]
\end{theorem}

\begin{proof}
Use Corollary \ref{C2} in combination with Lemma \ref{L3}.
\[
(v+1) B(n) ~\leq~ (v+1) B_\ell(n) +  \sum_{u=1}^{v} \left( u (n\alpha - B(n) ) - \binom{\ell}{2} \beta \right).
\]
\end{proof}

For linear regenerating codes with $n=k+1=d+1$, and for $\ell=n$, the bound 
\[
\binom{v+2}{2} B \leq \binom{v+1}{2} n \alpha + \binom{n}{2} \beta
\]
is Theorem 1.1 in \cite{PK15}. 

\bigskip


For a regenerating code with $(k=2p,d=3p)$ and $(\alpha=2p, \beta=1)$, the functional repair outer bound yields $B \leq (7p^2+p)/2$. The minimum in $B \leq \min \{ B_\ell(k) : 0 \leq \ell \leq k \}$ is attained for
$\ell=p$. Corollary 3.3 in \cite{D14} lowers the bound for exact repair regenerating codes 
by $(p^2-1)/16.$ Theorem \ref{T3}, with $\ell=n=k=2p$ and $v=1$, lowers the same bound by $p^2/6.$




\nocite{AhmadW14,DRWS11}

\newpage


\appendix

\section{Three proofs for $(4,3,3)$ outer bounds} \label{A433}

Proofs 1 and 2 are based on \cite{D14}. Proof 3 follows the current paper.

\bigskip

(Proof 1)
\[
 \begin {array}{ccccccc} 
&1&2&3&4\\ [1ex]
1~&\cd&\xd&\cd&\cd \\
2~&\cd&\cd&\cd&\cd \\
3~&\cd&\cd&\xa&\cd \\ 
4~&\cd&\cd&\cd&\xa 
  \end {array}  \qquad \qquad  \begin {array}{ccccccc} 
&1&2&3&4\\ [1ex]
1~&\xa&\cd&\cd&\cd \\
2~&\cd&\cd&\xd&\xd \\
3~&\cd&\cd&\cd&\xd \\ 
4~&\cd&\cd&\cd&\cd 
  \end {array}  \qquad \qquad   \begin {array}{ccccccc} 
&1&2&3&4\\ [1ex]
1~&\cd&\cd&\xd&\xd \\
2~&\cd&\xa&\cd&\cd \\
3~&\cd&\cd&\cd&\xd \\ 
4~&\cd&\cd&\cd&\cd 
  \end {array} 
\]
\begin{align*}
B &\leq H(W_4 | W_3 S_{1 \rightarrow 2}) +  H(W_3 S_{1 \rightarrow 2})  \\
                                            &\leq H(W_4 | S_{3 \rightarrow 4}) + H(W_3 S_{1 \rightarrow 2}) \\[1ex]
B &\leq H(S_{2 \rightarrow 3} | W_1 S_{2 \rightarrow 4} S_{3 \rightarrow 4}) + H(W_1 S_{2 \rightarrow 4} S_{3 \rightarrow 4}) \\
&\leq H(S_{2 \rightarrow 3} | W_4 S_{3 \rightarrow 4}) + H(W_1 S_{2 \rightarrow 4})  + H(S_{3 \rightarrow 4}) \\[1ex]
B &\leq H(S_{1 \rightarrow 3} |  W_2 S_{1 \rightarrow 4} S_{3 \rightarrow 4}) + H(W_2 S_{1 \rightarrow 4} S_{3 \rightarrow 4} )
 \\
&\leq H(S_{1 \rightarrow 3} |  S_{2 \rightarrow 3} W_4 S_{3 \rightarrow 4}) + H(W_2 S_{1 \rightarrow 4} S_{3 \rightarrow 4} )  \\[2ex]
3B &\leq  H(S_{1 \rightarrow 3} {S_{2 \rightarrow 3}} W_4 {S_{3 \rightarrow 4}})+ H(W_3 S_{1 \rightarrow 2}) + 
H(W_1 S_{2 \rightarrow 4}) + H(W_2 S_{1 \rightarrow 4} S_{3 \rightarrow 4} ) \\
&=  H(S_{1 \rightarrow 3} {S_{2 \rightarrow 3}} W_4 )+ H(W_3 S_{1 \rightarrow 2}) + 
H(W_1 S_{2 \rightarrow 4}) + H(W_2 S_{1 \rightarrow 4} S_{3 \rightarrow 4} ) 
 \\[2ex]
3B &\leq \sum_{1 \leq j \leq 4} H(W_j) + \sum_{1 \leq i < j \leq 4} H(\Sij).
\end{align*}

\bigskip

(Proof 2)
\[
 \begin {array}{ccccccc} 
&1&2&3&4\\ [1ex]
1~&\cd&\xd&\cd&\cd \\
2~&\cd&\cd&\cd&\cd \\
3~&\cd&\cd&\xa&\cd \\ 
4~&\cd&\cd&\cd&\xa 
  \end {array}  \qquad \qquad   \begin {array}{ccccccc} 
&1&2&3&4\\ [1ex]
1~&\xa&\cd&\cd&\cd \\
2~&\cd&\cd&\xd&\xd \\
3~&\cd&\cd&\cd&\xd \\ 
4~&\cd&\cd&\cd&\cd 
  \end {array} \qquad \qquad  \begin {array}{ccccccc} 
&1&2&3&4\\ [1ex]
1~&\cd&\cd&\cd&\cd \\
2~&\cd&\xa&\cd&\cd \\
3~&\cd&\cd&\xa&\cd \\ 
4~&\cd&\cd&\cd&\xa 
  \end {array} 
\]

\begin{align*}
B &\leq H(W_3 W_4) + H(S_{1 \rightarrow 2})  \\[1ex]
B &\leq H( S_{2 \rightarrow 3} S_{2 \rightarrow 4}) +  H( W_1 S_{3 \rightarrow 4}) \\[1ex]
B &\leq H(W_3 W_4 | W_2 ) + H(W_2) \leq H(W_3 W_4 | S_{2 \rightarrow 3} S_{2 \rightarrow 4}) + H(W_2)  \\[2ex]
3B &\leq H(W_3 W_4 S_{2 \rightarrow 3} S_{2 \rightarrow 4}) + H(W_3 W_4) + H(S_{1 \rightarrow 2})  + H( W_1 S_{3 \rightarrow 4}) + H(W_2) \\
&\leq H(W_3 W_4 S_{2 \rightarrow 3}) + H(W_3 W_4  S_{2 \rightarrow 4}) + H(S_{1 \rightarrow 2})  + H( W_1 S_{3 \rightarrow 4}) + H(W_2) \\
&\leq H(W_4 S_{1 \rightarrow 3} S_{2 \rightarrow 3}) + H( W_3 S_{1 \rightarrow 4} S_{2 \rightarrow 4})+ H(S_{1 \rightarrow 2})  + H( W_1 S_{3 \rightarrow 4}) + H(W_2) \\[2ex]
3B&\leq \sum_{1 \leq j \leq 4} H(W_j) + \sum_{1 \leq i < j \leq 4} H(\Sij).
\end{align*}

\newpage

(Proof 3)

\begin{table}[H]
\begin{center}
\begin{tabular}{cc@{\hspace{8mm}}ccccc@{\hspace{8mm}}ccccc}
\toprule 
&&1&2&3&4& &1&2&3&4\\[1ex]
\midrule
$M$ & &0  &0 &0 &0 & &$W_1$  &$W_2$ &$W_3$ &$W_4$ \\[1ex]
\midrule
$W_1$ & &$W_1$&$S_{1 \rightarrow 2}$&$S_{1 \rightarrow 3}$&$S_{1 \rightarrow 4}$ &
&$0$&$S_{1 \rightarrow 2}$&$S_{1 \rightarrow 3}$&$S_{1 \rightarrow 4}$ \\[1ex]
$W_2$ & &$W_2$&$W_2$&$S_{2 \rightarrow 3}$&$S_{2 \rightarrow 4}$ & &$W_2$&$0$&$S_{2 \rightarrow 3}$&$S_{2 \rightarrow 4}$  \\[1ex]
$W_3$ & &$W_3$&$W_3$&$W_3$&$S_{3 \rightarrow 4}$& &$W_3$&$W_3$&0&$S_{3 \rightarrow 4}$  \\[1ex] 
$W_4$ & &$W_4$&$W_4$&$W_4$&$W_4$ & &$W_4$&$W_4$&$W_4$&$0$ \\[1ex]
\midrule
$W_5$ & &$W_5$&$W_5$&$W_5$&$W_5$   & &$W_5$&$W_5$&$W_5$&$W_5$   \\[1ex]
\bottomrule 
\end{tabular}
\end{center}
\caption{Pairs of columns with the same entropy}
\label{Table:marg2}
\end{table}


Table \ref{Table:marg2} contains random variables $W_i$ and $\Sij$ for a regenerating code with four nodes and parameters $(k=3,d=3)$. The last row is obtained by adding a node $W_5$ whose contents will be chosen later.
Columns with the same label $1 \leq j \leq 4$ contain the same variables. For a pair of columns with the same label we compute the column entropy using the chain rule from the bottom to the top. The computation is done first for columns from row $W_4$ upwards and then for the extended columns from row $W_5$ upwards. By invoking the chain rule each entry in the table contributes to the entropy of its column with its entropy conditional on the entries below it. 

\bigskip

We compare the sum of the column entropies for the four columns on the left and on the right. First we ignore the row $W_5$ (or set $W_5=0$). The entries below the diagonal produce the same terms left and right. The four remaining entries with $W_i$ on the left (in the diagonal positions) sum to $H(W_1 W_2 W_3 W_4)$. The four remaining entries with $W_i$ on the right (in the top row) all produce $0$ terms (using $d=3$). The remaining entries on the right sum to at most $\sum_{1 \leq i < j \leq 4} S_{i \rightarrow j}$ (with equality if and only if for each term the conditional entropy equals the actual entropy). We still have to account for the entries with $S_{i \rightarrow j}$ on the left. In each case, an entry contributes at least $H(\Sij | W_{[i+1,4]})$. Thus
 \[
H(W_1 W_2 W_3 W_4) +  \sum_{1 \leq i < j \leq 4} H(\Sij|W_{[i+1,4]}) \leq \sum_{1 \leq i < j \leq 4} H(\Sij).
\]
We repeat the comparison but now include the constant row $W_5$.
\[
H(W_1 W_2 W_3 W_4| W_5 ) +  \sum_{1 \leq i < j \leq 4} H(\Sij|W_{[i+1,5]}) \leq \sum_{1 \leq i < j \leq 4} H(\Sij| W_5).
\]
For $W_5$ such that
\[
H(\Sij|W_5) \leq H(\Sij|W_{[i+1,4]})
\]
we obtain 
\begin{align*}
&H(W_1 W_2 W_3 W_4 ) + H(W_1 W_2 W_3 W_4 | W_5 )  \\ \leq~ &H(W_1 W_2 W_3 W_4 ) + \sum_{1 \leq i < j \leq 4} H(\Sij|W_5) 
~\leq~ \sum_{1 \leq i < j \leq 4} H(\Sij).
\end{align*}
In the linear setting it suffices to choose for $W_5$ a vector space that  contains $W_j \cap W_{[j+1,4]}$ for  $j = 1,2,3$.
This results in $H(W_5) = \sum_{1 \leq j \leq 4} H(W_j) - H(W_1 W_2 W_3 W_4)$. And, with $B= H(W_1 W_2 W_3 W_4 )$, in
\[
3B \leq \sum_{1 \leq j \leq 4} H(W_j) + \sum_{1 \leq i < j \leq 4} H(\Sij).
\]

\end{document}